\documentclass[onecolumn, a4size, 12pt]{IEEEtran}
\usepackage{amsmath}
\usepackage{amssymb}
\usepackage{amsfonts}
\usepackage{graphicx}
\usepackage{epsfig}
\usepackage{subfigure}
\usepackage{psfrag}

\title{Exploiting Opportunistic Multiuser Detection in Decentralized
Multiuser MIMO Systems \footnote{R. Zhang is with the Institute for
Infocomm Research, A*STAR, Singapore
(e-mail:rzhang@i2r.a-star.edu.sg).} \footnote{J. M. Cioffi is with
the Department of Electrical Engineering, Stanford University,
Stanford, USA (e-mail:cioffi@stanford.edu).}}

\author{Rui Zhang and John M. Cioffi}

\begin{document}
\maketitle \thispagestyle{empty}

\begin{abstract}
This paper studies the design of a \emph{decentralized} multiuser
multi-antenna (MIMO) system for spectrum sharing over a fixed narrow
band, where the coexisting users independently update their transmit
covariance matrices for individual transmit-rate maximization via an
iterative manner. This design problem was usually investigated in
the literature by assuming that each user treats the co-channel
interference from all the other users as additional (colored) noise
at the receiver, i.e., the conventional \emph{single-user decoder}
(SUD) is applied. This paper proposes a new decoding method for the
decentralized multiuser MIMO system, whereby each user
opportunistically cancels the co-channel interference from some or
all of the other users via applying multiuser detection techniques,
thus termed \emph{opportunistic multiuser detection} (OMD). This
paper studies the optimal transmit covariance design for users'
iterative maximization of individual transmit rates with the
proposed OMD, and demonstrates the resulting capacity gains in
decentralized multiuser MIMO systems against the conventional SUD.
\end{abstract}


\setlength{\baselineskip}{1.3\baselineskip}
\newtheorem{definition}{\underline{Definition}}[section]
\newtheorem{fact}{Fact}
\newtheorem{assumption}{Assumption}
\newtheorem{theorem}{\underline{Theorem}}[section]
\newtheorem{lemma}{\underline{Lemma}}[section]
\newtheorem{corollary}{Corollary}
\newtheorem{proposition}{\underline{Proposition}}[section]
\newtheorem{example}{\underline{Example}}[section]
\newtheorem{remark}{\underline{Remark}}[section]
\newtheorem{algorithm}{\underline{Algorithm}}[section]
\newcommand{\mv}[1]{\mbox{\boldmath{$ #1 $}}}

\section{Introduction}

The Gaussian interference channel is a basic mathematical model that
characterizes many real-life communication systems with multiple
uncoordinated users sharing a common spectrum to transmit
independent information at the same time, such as the digital
subscriber line (DSL) network \cite{Cioffi}, the ad-hoc wireless
network \cite{Andrews}, and the newly emerging cognitive radio (CR)
wireless network \cite{Tarokh}. From an information-theoretical
perspective, the capacity region of the Gaussian interference
channel, which constitutes all the simultaneously achievable rates
of the users in the system, is still unknown in general \cite{Han},
while significant progresses have recently been made on approaching
this limit \cite{Tse}, \cite{Jafar}. Capacity-approaching techniques
usually require certain cooperations among distributed users for
their encoding and decoding. A more pragmatic approach that leads to
suboptimal achievable rates of the users in the Gaussian
interference channel is to restrict the system to operate in a
decentralized manner \cite{Yu02}, i.e., allowing only single-user
encoding and decoding by treating the co-channel interference from
the other users as additional Gaussian noise at each user's
receiver. In such a context, decentralized algorithms for users to
allocate their transmit resources such as the power, bit-rate,
bandwidth, and antenna beam to optimize individual transmission
performance and yet to ensure certain fairness among all the users,
become most important.

This paper focuses on a multiuser multiple-input multiple-output
(MU-MIMO) wireless system, where multiple distributed links, each
equipped with multiple transmit and/or receive antennas, share a
common narrow band for transmission in a fully decentralized manner.
In such a scenario, the system design reduces to finding a set of
transmit covariance matrices for the users subject to their
co-channel interference resulting from their simultaneous and
uncoordinated transmissions. This design problem has been
investigated in a vast number of prior works in the literature,
e.g., \cite{Ingram}-\cite{Palomar}, by treating the co-channel
interference as additional colored noise at each user's receiver,
i.e., the conventional {\it single-user decoder} (SUD) for the
classic point-to-point MIMO channel is applied. In \cite{Ingram},
the authors proposed an algorithm, which is in spirit analogous to
the iterative water-filling (IWF) algorithm in \cite{Yu02}, for each
distributed MIMO link to iteratively update transmit covariance
matrix to maximize individual transmit rate. Distributed iterative
beamforming (the rank of transmit covariance matrix is restricted to
be one) algorithms were also studied in \cite{Tassiulas} for
transmit sum-power minimization given individual user's quality of
service (QoS) constraint in terms of the received
signal-to-interference-plus-noise ratio (SINR). The throughput of
decentralized MU-MIMO systems has been further analyzed in
\cite{Blum03a} and \cite{Chen06} for the cases of fading channels
and large-size systems, respectively. In \cite{Blum03b},
\cite{Larsson}, centralized strategies were proposed where all
users' transmit covariance matrices are jointly searched to maximize
their sum-rate, and numerical algorithms were also proposed to
converge to a local sum-rate maxima. Analyzing the decentralized
MU-MIMO system via a game theoretical approach has recently been
done in \cite{Liang}-\cite{Palomar}.

The cited papers on decentralized/centralized designs for the
Gaussian MIMO interference channel have all adopted the SUD at each
user's receiver, whereas during the past decade multiuser detection
techniques (see, e.g., \cite{Verdu} and references therein) have
been thoroughly investigated in the literature, and have been proven
in realistic multiuser/MIMO systems to be able to provide
substantial performance gains over the conventional SUD. This
motivates our work's investigation of the following question:
Considering a decentralized MU-MIMO system where the users
iteratively adapt their transmit covariance matrices for individual
rate maximization, ``Is applying multiuser detection at each user's
receiver able to enhance the system throughput over the conventional
SUD?'' Note that because of the randomness of channels among the
users, as well as their independent rate assignments, at one
particular user's receiver, multiuser detection can be used to
cancel the co-channel interference from some/all of its coexisting
users only when their received signals are jointly decodable with
this particular user's own received signal. Thus, we refer to this
decoding method as {\it opportunistic multiuser detection} (OMD).
Also note that the OMD in the context of the decentralized MU-MIMO
system is analogous to the ``successive group decoder (SGD)'' in the
fading multiple-access channel (MAC) with unknown channel state
information (CSI) at the user transmitters (see, e.g., \cite{Wang08}
and references therein). With the proposed OMD, this paper derives
the optimal transmit covariance matrix for user's individual
transmit-rate maximization at each iteration of transmit adaptation.
By simulation, this paper demonstrates the throughput gains of the
converged users' transmit covariance matrices with the proposed OMD
over the conventional SUD.

The rest of this paper is organized as follows. Section
\ref{sec:system model} presents the system model of the
decentralized MU-MIMO system. Section \ref{sec:2 users} studies the
optimal design of user transmit covariance matrix with the proposed
OMD for the special case with two users in the system. Section
\ref{sec:K users} generalizes the results to the case of more than
two users. Section \ref{sec:numerical results} provides the
simulation results to demonstrate the throughput gains with the
proposed OMD over the SUD. Finally, Section \ref{sec:conclusion}
concludes the paper.

{\it Notation}: Scalars are denoted by lower-case letters, e.g.,
$x$, and bold-face lower-case letters are used for vectors, e.g.,
$\mv{x}$, and bold-face upper-case letters for matrices, e.g.,
$\mv{X}$. In addition, $\mathtt{tr}(\mv{S})$, $|\mv{S}|$,
$\mv{S}^{-1}$, and $\mv{S}^{\frac{1}{2}}$ denote the trace,
determinant, inverse, and square-root of a square matrix $\mv{S}$,
respectively, and $\mv{S}\succeq 0$ means that $\mv{S}$ is a
positive semi-definite matrix \cite{Boyd}. For an arbitrary-size
matrix $\mv{M}$, $\mv{M}^{H}$ denotes the conjugate transpose of
$\mv{M}$. $\mathtt{diag}(x_1, \ldots, x_M)$ denotes a $M \times M$
diagonal matrix with $x_1,\ldots,x_M$ as its diagonal elements.
$\mv{I}$ and $\mv{0}$ denote the identity matrix and the all-zero
vector, respectively. $\mathbb{E}[\cdot]$ denotes the statistical
expectation. The distribution of a circular symmetric complex
Gaussian (CSCG) random vector with mean $\mv{x}$ and covariance
matrix $\mv{\Sigma}$ is denoted by
$\mathcal{CN}(\mv{x},\mv{\Sigma})$, and $\sim$ stands for
``distributed as''. $\mathbb{C}^{x \times y}$ denotes the space of
$x\times y$ matrices with complex-valued elements. $\max(x,y)$ and
$\min(x,y)$ denote the maximum and minimum between two real numbers,
$x$ and $y$, respectively, and $(x)^+=\max(x,0)$.

\section{System Model}\label{sec:system model}

This paper considers a distributed MU-MIMO system where $K$ users
transmit independent information to their corresponding receivers
simultaneously over a common narrow band. Each user is equipped with
multiple transmit and/or receiver antennas, while for user $k$,
$k=1,\ldots,K$, $N_k$ and $M_k$ denote the number of its transmit
and receive antennas, respectively. For the time being, it is
assumed that perfect time and frequency synchronization with
reference to a common clock system have been established for all the
users in the system prior to data transmission. We also assume a
{\it block-fading} model for all the channels involved in the
system, and a block-based transmission for all the users over each
particular channel fading state. Since the proposed study applies to
any channel fading state, for brevity we drop the index of fading
state here. The discrete-time baseband signal for the $k$th user
transmission is given by
\begin{align}\label{eq:signal model}
\mv{y}_{k}=\mv{H}_{kk}\mv{x}_k+\sum_{j=1,j\neq
k}^K\mv{H}_{jk}\mv{x}_j+\mv{z}_k
\end{align}
where $\mv{x}_k\in\mathbb{C}^{N_k\times 1}$ and
$\mv{y}_k\in\mathbb{C}^{M_k\times 1}$ are the transmitted and
received signal vectors for user $k$, respectively,
$k\in\{1,\ldots,K\}$; $\mv{H}_{kk}\in\mathbb{C}^{M_k\times N_k}$
denotes the direct-link channel matrix for user $k$, while
$\mv{H}_{jk}\in\mathbb{C}^{M_k\times N_j}$ denotes the cross-link
channel matrix from user $j$ to user $k$, $j\in\{1,\ldots,K\}$,
$j\neq k$; and  $\mv{z}_k\in\mathbb{C}^{M_k \times 1}$ is the
received noise vector of user $k$.

Without loss of generality, it is assumed that
$\mv{z}_k\sim\mathcal{CN}(\mv{0},\mv{I}), \forall
k\in\{1,\ldots,K\}$, and all $\mv{z}_k$'s are independent. We
consider a decentralized multiuser system where the $K$ users
independently encode their transmitted messages and thus
$\mv{x}_k$'s are independent over $k$. Since this paper is
interested in the information-theoretic limit of each Gaussian MIMO
channel involved, it is assumed that
$\mv{x}_k\sim\mathcal{CN}(\mv{0},\mv{S}_k), \forall
k\in\{1,\ldots,K\}$, where $\mv{S}_k=\mathbb{E}[\mv{x}_k\mv{x}_k^H]$
is the transmit covariance matrix for user $k$.

This paper considers a similar decentralized operation protocol as
in \cite{Yu02}, \cite{Ingram}, \cite{Liang}-\cite{Palomar}, whereby
the users in the system take turns to update their transmit
covariance matrices for individual rate maximization, with all the
other users' transmit covariance matrices being fixed, until all
users' transmit covariance matrices and their transmit rates get
converged. We consider two types of decoding methods at each user's
receiver. One is the conventional SUD, which has been applied in the
above cited papers, where the $k$th user decodes its desired message
by treating the co-channel interference from all the other users,
$j\neq k$, as additional colored Gaussian noise
$\sim\mathcal{CN}(\mv{0},\sum_{j=1,j\neq
k}^K\mv{H}_{jk}\mv{S}_j\mv{H}_{jk}^H)$. The other decoding method is
the newly proposed OMD, whereby each user opportunistically applies
multiuser detection to decode some/all of its coexisting users'
messages so as to cancel their resulted interference, provided that
these messages are jointly decodable with this user's own message.
In practice, each user in the system is usually interfered with by
all the other users, while due to location-dependent
shadowing/fading, only a small group of coexisting users who are
closest to one particular user and thus correspond to the strongest
cross-link channels to this user, will contribute the most to this
user's received co-channel interference. As a result, this user can
effectively estimate the transmit rates as well as the cross-link
channels of these ``strong'' interference users, and employ the
proposed OMD to suppress their interference at the receiver. Note
that the use of OMD instead of SUD still maintains the fully
decentralized property of the existing IWF-like operation protocols
given in \cite{Yu02}, \cite{Ingram}, \cite{Liang}-\cite{Palomar}.

\section{Transmit Covariance Optimization: The Two-User Case} \label{sec:2 users}

In this section, we present the problem formulation as well as the
solution to determine the optimal transmit covariance matrix of each
user for individual transmit-rate maximization, when the proposed
OMD is employed. For the purpose of exposition, we consider the
special case where only two users exist in the system. We will
address the general case with more than two users in Section
\ref{sec:K users}. For brevity, only user 1's transmit adaptation is
addressed here, while the developed results apply similarly to user
2.

\subsection{Problem Formulation}

Note that at one particular iteration of user 1 to update its
transmission, user 2's transmit covariance matrix, $\mv{S}_2$, and
transmit rate, denoted by $r_2$, are both fixed values. For a given
transmit covariance matrix of user 1, $\mv{S}_1$, the resultant
maximum transmit rate of user 1 can be expressed as
\begin{equation}\label{eq:rate}
r_1(\mv{S}_1)=\left\{
\begin{array}{ll} \log\left|\mv{I}+\mv{H}_{11}\mv{S}_1\mv{H}_{11}^H\right| &
r_2\leq R_2^{(a)}  \\
\log\left|\mv{I}+\mv{H}_{11}\mv{S}_1\mv{H}_{11}^H+\mv{H}_{21}\mv{S}_2\mv{H}_{21}^H\right|-r_2
& R_2^{(a)}<r_2\leq R_2^{(b)}
\\ \log\left|\mv{I}+(\mv{I}+\mv{H}_{21}\mv{S}_2\mv{H}_{21}^H)^{-1}\mv{H}_{11}\mv{S}_1\mv{H}_{11}^H\right| & r_2> R_2^{(b)}
\end{array} \right.
\end{equation}
where
\begin{align}
R_2^{(a)}&=\log\left|\mv{I}+(\mv{I}+\mv{H}_{11}\mv{S}_1\mv{H}_{11}^H)^{-1}\mv{H}_{21}\mv{S}_2\mv{H}_{21}^H\right|
\label{eq:R(1)}
\\
R_2^{(b)}&=\log\left|\mv{I}+\mv{H}_{21}\mv{S}_2\mv{H}_{21}^H\right|.
\label{eq:R(2)}
\end{align}

The above result is illustrated in the following three cases
corresponding to the three expressions of $r_1$ in (\ref{eq:rate})
from top to bottom.

\begin{itemize}
\item {\it Strong Interference Case}: In this case, the received
signal from user 2 is decodable at user 1's receiver with the
conventional SUD, by treating user 1's signal as colored Gaussian
noise. This is feasible since $r_2\leq R_2^{(a)}$ given in
(\ref{eq:R(1)}). After decoding user'2 message and thereby canceling
its associated interference, user 1 can decode its own message with
a maximum rate equal to its own channel capacity. The above decoding
method is known as {\it successive decoding} (SD) for the standard
Gaussian MAC \cite{Cover}.

\item {\it Moderate Interference Case}: In this case, $r_2> R_2^{(a)}$ and thus the received signal
from user 2 is not directly decodable by the SUD. However, since
$r_2\leq R_2^{(b)}$ given in (\ref{eq:R(2)}), it is still feasible
for user 1 to apply {\it joint decoding} (JD) \cite{Cover} to decode
both users' messages.\footnote{Note that SD can also be applied in
this case to achieve the same rate for user 1 as JD, if SD is
deployed jointly with the ``time sharing'' \cite{Cover} or ``rate
splitting'' \cite{Rimoldi} encoding technique at user 1's
transmitter. Since these techniques require certain cooperations
between users, they might not be suitable for the fully
decentralized multiuser system considered in this paper.} In this
case, the rate pair of the two users should lie on the $45$-degree
segment of the corresponding MAC capacity region boundary
\cite{Cover}, i.e.,
$r_1+r_2=\log\left|\mv{I}+\mv{H}_{11}\mv{S}_1\mv{H}_{11}^H+\mv{H}_{21}\mv{S}_2\mv{H}_{21}^H\right|$.

\item {\it Weak Interference Case}: In this case, $r_2> R_2^{(b)}$, i.e., the received signal from user 2 is not decodable even
without the presence of user 1's signal. As such, user 1's receiver
has the only option of treating user 2's signal as colored Gaussian
noise and applying the conventional SUD to directly decode user 1's
message, the same as that in the existing IWF-like algorithms (see,
e.g., \cite{Ingram}, \cite{Liang}-\cite{Palomar}).
\end{itemize}

In the above decoding method, multiuser detection is applied in both
cases of strong and moderate interferences when $r_2\leq R_2^{(b)}$,
but not in the case of weak interference when $r_2> R_2^{(b)}$.
Thus, user 1's receiver opportunistically applies multiuser
detection to decode user 2's message, either successively (SD) or
jointly (JD) with its own message. We thus refer to this decoding
method as {\it opportunistic multiuser detection} (OMD). From
(\ref{eq:R(1)}) and (\ref{eq:R(2)}), it follows that $R_2^{(a)}\leq
R_2^{(b)}$. Further more, it is easy to verify that $r_1$ given in
(\ref{eq:rate}) with the proposed OMD is in general larger than the
achievable rate with the conventional SUD (given by the third
expression of $r_1$ in  (\ref{eq:rate}) independent of $r_2$), for
any given set of $\mv{S}_1, \mv{S}_2$, and $r_2$.

With $r_1(\mv{S}_1)$ given in (\ref{eq:rate}) for a fixed
$\mv{S}_1$, we can further maximize user 1's transmit rate by
searching over $\mv{S}_1$. Let $P_1$ denote the transmit power
constraint of user 1. This problem can be expressed as
\begin{align}
\mbox{(P1)}~~\mathop{\mathtt{max}}_{\mv{S}_1} & ~~~ r_1(\mv{S}_1)
\nonumber \\
\mathtt{s.t.} & ~~~ \mathtt{tr}(\mv{S}_1)\leq P_1, \mv{S}_1\succeq 0
\nonumber
\end{align}
where $r_1(\mv{S}_1)$ is given in (\ref{eq:rate}). The optimal
solution of $\mv{S}_1$ in (P1) and the corresponding maximum
transmit rate of user 1 are denoted by $\mv{S}_1^{\rm OMD}$ and
$r_1^{\rm OMD}$, respectively.

\subsection{Proposed Solution}

In this subsection, we study the solution of (P1) for the optimal
transmit covariance matrix of user 1, when the proposed OMD is
deployed at user 1' receiver. Note that although the constraints of
(P1) are convex, its objective function is not necessarily concave
due to the fact that $R_2^{(a)}$ given in (\ref{eq:R(1)}) is neither
convex nor concave function of $\mv{S}_1$. As a result, (P1) seems
to be non-convex at a first glance. In fact, (P1) is a convex
optimization problem after being transformed into a convex form, as
will be shown in this subsection. In the following, we will study
the solution of (P1) for two cases: $r_2> R_2^{(b)}$ and $r_2\leq
R_2^{(b)}$, for which the SUD and the multiuser decoding (MD) (in
the form of either SD or JD) should be used to achieve
$r_1({\mv{S}_1})$ given in (\ref{eq:rate}), respectively.

\subsubsection{$r_2>R_2^{(b)}$}

In this case, the SUD should be applied. Note that $R_2^{(b)}$ is a
constant unrelated to $\mv{S}_1$. Thus, the optimal $\mv{S}_1$ that
maximizes the third expression of $r_1(\mv{S}_1)$ in (\ref{eq:rate})
has the following structure \cite{Cover}:
\begin{equation}\label{eq:optimal S SUD}
\mv{S}_1^{\rm SUD}=\mv{V}\mv{\Lambda}\mv{V}^H
\end{equation}
where $\mv{V}\in\mathbb{C}^{N_1\times T_1}$ with $T_1=\min(N_1,M_1)$
is obtained from the singular-value decomposition (SVD) of the
equivalent channel of user 1 (after the noise whitening) expressed
as
\begin{equation}
(\mv{I}+\mv{H}_{21}\mv{S}_2\mv{H}_{21}^H)^{-\frac{1}{2}}\mv{H}_{11}=\mv{U}\mv{\Sigma}\mv{V}^H
\end{equation}
with $\mv{U}\in\mathbb{C}^{M_1\times T_1}$, $\mv{\Sigma}=\mathtt{
diag}(\sigma_1,\ldots,\sigma_{T_1})$, $\sigma_i\geq 0$,
$i=1,\ldots,T_1$, and $\mv{\Lambda}=\mathtt{
diag}(p_1,\ldots,p_{T_1})$ with $p_i$'s obtained from the standard
water-filling solution \cite{Cover}:
\begin{equation}
p_i=\left(\mu-\frac{1}{\sigma_i^2}\right)^+, \ \ i=1,\ldots,T_1,
\end{equation}
with $\mu$ being a constant to make $\sum_{i=1}^{T_1}p_i=P_1$. The
maximum rate of user 1 then becomes
\begin{equation}\label{eq:optimal rate SUD}
r_1^{\rm SUD}=\sum_{i=1}^{T_1}\log(1+\sigma_i^2p_i).
\end{equation}

\subsubsection{$r_2\leq R_2^{(b)}$}

In this case, the MD in the form of either SD or JD should be used.
In order to overcome the non-concavity of $r_1(\mv{S}_1)$ given in
(\ref{eq:rate}) due to $R_2^{(a)}$, we re-express the first two
expressions of $r_1(\mv{S}_1)$ in (\ref{eq:rate}) as
\begin{align}\label{eq:rate new}
r_1^{\rm
MD}(\mv{S}_1)=\min\left(\log\left|\mv{I}+\mv{H}_{11}\mv{S}_1\mv{H}_{11}^H\right|,
\log\left|\mv{I}+\mv{H}_{11}\mv{S}_1\mv{H}_{11}^H+\mv{H}_{21}\mv{S}_2\mv{H}_{21}^H\right|-r_2\right).
\end{align}

Thus, the maximum achievable rate of user 1 can be obtained as
\begin{equation}\label{eq:rate MD}
r_1^{\rm MD}=\max_{\mv{S}_1:\mathtt{tr}(\mv{S}_1)\leq P_1,
\mv{S}_1\succeq 0} r_1^{\rm MD}(\mv{S}_1).
\end{equation}
The maximization problem in (\ref{eq:rate MD}) can be explicitly
written as
\begin{align}
\mbox{(P2)}~~\mathop{\mathtt{max}}_{r_1, ~ \mv{S}_1} & ~~~ r_1
\nonumber \\
\mathtt{s.t.} & ~~~ r_1\leq
\log\left|\mv{I}+\mv{H}_{11}\mv{S}_1\mv{H}_{11}^H\right| \label{eq:rate constraint 1} \\
& ~~~ r_1\leq
\log\left|\mv{I}+\mv{H}_{11}\mv{S}_1\mv{H}_{11}^H+\mv{H}_{21}\mv{S}_2\mv{H}_{21}^H\right|-r_2
\label{eq:rate constraint 2} \\ & ~~~ r_1\geq 0,
\mathtt{tr}(\mv{S}_1)\leq P_1, \mv{S}_1\succeq 0. \label{eq:other
constraints}
\end{align}
The optimal solution of $r_1$ in (P2) will be $r_1^{\rm MD}$. Note
that (P2) is a convex optimization problem since its constraints
specify a convex set of $(r_1,\mv{S}_1)$. To solve (P2), we apply
the standard Lagrange duality method \cite{Boyd}. First, we
introduce two non-negative dual variables, $\mu_1$ and $\mu_2$,
associated with the two rate constraints (\ref{eq:rate constraint
1}) and (\ref{eq:rate constraint 2}), respectively, and write the
associated Lagrangian of (P2) as
\begin{align}\label{eq:Lagrangian}
\mathcal{L}(r_1,\mv{S}_1,\mu_1,\mu_2)=&r_1-\mu_1\left(r_1-\log\left|\mv{I}+\mv{H}_{11}\mv{S}_1\mv{H}_{11}^H\right|\right)
\nonumber
\\ & - \mu_2\left(r_1-\log\left|\mv{I}+\mv{H}_{11}\mv{S}_1\mv{H}_{11}^H+\mv{H}_{21}\mv{S}_2\mv{H}_{21}^H\right|+r_2\right)
\end{align}
By reordering the terms in (\ref{eq:Lagrangian}), we obtain
\begin{align}\label{eq:Lagrangian new}
\mathcal{L}(r_1,\mv{S}_1,\mu_1,\mu_2)=&(1-\mu_1-\mu_2)r_1
+\mu_1\log\left|\mv{I}+\mv{H}_{11}\mv{S}_1\mv{H}_{11}^H\right|
\nonumber \\ & +
\mu_2\log\left|\mv{I}+\mv{H}_{11}\mv{S}_1\mv{H}_{11}^H+\mv{H}_{21}\mv{S}_2\mv{H}_{21}^H\right|+\mu_2r_2.
\end{align}
The Lagrange dual function of (P2) is then defined as
\begin{eqnarray}\label{eq:Lagrange dual}
g(\mu_1,\mu_2)=\max_{(r_1,\mv{S}_1)\in\mathcal{A}}\mathcal{L}(r_1,\mv{S}_1,\mu_1,\mu_2)
\end{eqnarray}
where the set $\mathcal{A}$ specifies the remaining constraints of
(P2) given in (\ref{eq:other constraints}). The dual problem of
(P2), of which the optimal value is the same as that of
(P2),\footnote{It can be easily checked that the Slater's condition
holds for (P2) and thus the duality gap for (P2) is zero
\cite{Boyd}.} is defined as
\begin{eqnarray}\label{eq:dual problem}
\mbox{(P2-D)}~~\min_{\mu_1\geq 0,\mu_2\geq 0} g(\mu_1,\mu_2).
\end{eqnarray}

Let $r_1^*$ and $\mv{S}_1^*$ denote the optimal solutions of (P2).
Let $\mu_1^*$ and $\mu_2^*$ denote the optimal dual solutions of the
dual problem (P2-D). Next, we will present a key relationship
between $\mu_1^*$ and $\mu_2^*$ as follows.
\begin{lemma}\label{lemma:optimal mu}
In problem (P2-D), the optimal solutions satisfy that
$\mu_1^*+\mu_2^*=1$.
\end{lemma}
\begin{proof}
See Appendix \ref{appendix:proof optimal mu}.
\end{proof}

Given Lemma \ref{lemma:optimal mu}, without loss of generality, we
can replace $\mu_2$ by $1-\mu_1$ in (\ref{eq:Lagrangian new}). Thus,
the maximization problem in (\ref{eq:Lagrange dual}) can be
equivalently rewritten as (by discarding the constant term
$\mu_2r_2$)
\begin{align}
\mbox{(P3)}~~\mathop{\mathtt{max}}_{\mv{S}_1} & ~~~
\mu_1\log\left|\mv{I}+\mv{H}_{11}\mv{S}_1\mv{H}_{11}^H\right|+
(1-\mu_1)\log\left|\mv{I}+\mv{H}_{11}\mv{S}_1\mv{H}_{11}^H+\mv{H}_{21}\mv{S}_2\mv{H}_{21}^H\right|
\nonumber \\
\mathtt{s.t.} & ~~~ \mathtt{tr}(\mv{S}_1)\leq P_1, \mv{S}_1\succeq
0.
\end{align}

Further more, the dual problem (\ref{eq:dual problem}) now only
needs to minimize $g(\mu_1)$ (since $\mu_2=1-\mu_1$) over
$0\leq\mu_1\leq 1$. Then, there are the following three cases in
which $\mu_1^*$ takes different values.
\begin{itemize}
\item $\mu_1^*=0$: In this case, $\mu_2^*=1$. From the Karush-Kuhn-Tucker (KKT) optimality conditions
\cite{Boyd} of (P2), it is known that the constraint (\ref{eq:rate
constraint 1}) is inactive while the constraint (\ref{eq:rate
constraint 2}) is active. This suggests that JD instead of SD is
optimal. Furthermore, from (P3), with $\mu_1=\mu_1^*=0$, it follows
that $\mv{S}_1^*$, denoted by $\mv{S}_1^{\rm JD}$, maximizes the
sum-rate,
$\log\left|\mv{I}+\mv{H}_{11}\mv{S}_1\mv{H}_{11}^H+\mv{H}_{21}\mv{S}_2\mv{H}_{21}^H\right|$,
from which we can show that
\begin{equation}\label{eq:optimal S JD}
\mv{S}_1^{\rm JD}=\mv{S}_1^{\rm SUD}
\end{equation}
where $\mv{S}_1^{\rm SUD}$ is given in (\ref{eq:optimal S SUD}),
i.e., the optimal transmit covariance matrix is the same for both
cases of SUD and JD. However, the optimal $r_1^*$ in this case with
JD, denoted by $r_1^{\rm JD}$, is equal to
\begin{equation}\label{eq:optimal rate JD}
r_1^{\rm JD}=r_1^{\rm SUD}+R_2^{(b)}-r_2
\end{equation}
where $r_1^{\rm SUD}$ is given in (\ref{eq:optimal rate SUD}).
Finally, we need to check the condition under which this case holds.
Since the constraint (\ref{eq:rate constraint 1}) should be
inactive, it follows that
\begin{equation}\label{eq:condition JD}
r_1^{\rm JD}<\log\left|\mv{I}+\mv{H}_{11}\mv{S}_1^{\rm
JD}\mv{H}_{11}^H\right|.
\end{equation}
From (\ref{eq:optimal rate JD}) and (\ref{eq:condition JD}), it can
be shown that the case of interest holds when
\begin{equation}\label{eq:R(1) bar}
r_2>\log\left|\mv{I}+(\mv{I}+\mv{H}_{11}\mv{S}_1^{\rm
JD}\mv{H}_{11}^H)^{-1}\mv{H}_{21}\mv{S}_2\mv{H}_{21}^H
\right|\triangleq \bar{R}_2^{(a)}.
\end{equation}
Note that $\bar{R}_2^{(a)}$ can also be obtained from $R_2^{(a)}$
given in (\ref{eq:R(1)}) by letting $\mv{S}_1=\mv{S}_1^{\rm JD}$.

\item $\mu_1^*=1$: In this case, $\mu_2^*=0$. From the KKT optimality conditions
of (P2), it is known that the constraint (\ref{eq:rate constraint
1}) is active while the constraint (\ref{eq:rate constraint 2}) is
inactive. This suggests that SD instead of JD is optimal.
Furthermore, from (P3), with $\mu_1=\mu_1^*=1$, it follows that
$\mv{S}_1^*$, denoted by $\mv{S}_1^{\rm SD}$, maximizes user 1's own
channel capacity (without the presence of user 2),
$\log\left|\mv{I}+\mv{H}_{11}\mv{S}_1\mv{H}_{11}^H\right|$, from
which we can easily show that \cite{Cover}
\begin{equation}\label{eq:optimal S SD}
\mv{S}_1^{\rm SD}=\mv{V}_1\mv{\Lambda}_1\mv{V}_1^H
\end{equation}
where $\mv{V}_1\in\mathbb{C}^{N_1\times T_1}$ is obtained from the
SVD of the direct-link channel of user 1 expressed as
$\mv{H}_{11}=\mv{U}_1\mv{\Gamma}\mv{V}_1^H$, with
$\mv{U}_1\in\mathbb{C}^{M_1\times T_1}$, $\mv{\Gamma}_1=\mathtt{
diag}(\gamma_1,\ldots,\gamma_{T_1})$, $\gamma_i\geq 0$,
$i=1,\ldots,T_1$, and $\mv{\Lambda}_1=\mathtt{
diag}(q_1,\ldots,q_{T_1})$ with $q_i$'s obtained from the standard
water-filling solution \cite{Cover}:
\begin{equation}
q_i=\left(\nu-\frac{1}{\gamma_i^2}\right)^+, \ \ i=1,\ldots,T_1,
\end{equation}
with $\nu$ being a constant to make $\sum_{i=1}^{T_1}q_i=P_1$. The
optimal $r_1^*$ in this case with SD, denoted by $r_1^{\rm SD}$,
then becomes
\begin{equation}\label{eq:optimal rate SD}
r_1^{\rm SD}=\sum_{i=1}^{T_1}\log(1+\gamma_i^2q_i).
\end{equation}
Similarly like the previous case, we can show that this case holds
when
\begin{equation}\label{eq:R(3)}
r_2<\log\left|\mv{I}+(\mv{I}+\mv{H}_{11}\mv{S}_1^{\rm
SD}\mv{H}_{11}^H)^{-1}\mv{H}_{21}\mv{S}_2\mv{H}_{21}^H
\right|\triangleq \hat{R}_2^{(a)}.
\end{equation}
At last, we have the following lemma.
\begin{lemma}\label{lemma:inequality}
For $\bar{R}_2^{(a)}$ defined in (\ref{eq:R(1) bar}) and
$\hat{R}_2^{(a)}$ defined in (\ref{eq:R(3)}), it holds that
$\bar{R}_2^{(a)}\geq \hat{R}_2^{(a)}$.
\end{lemma}
\begin{proof}
See Appendix \ref{appendix:proof inequality}.
\end{proof}

\item $0<\mu_1^*<1$: In this case, $0<\mu_2^*<1$, and from the KKT optimality conditions
of (P2), it is known that both the constraints (\ref{eq:rate
constraint 1}) and (\ref{eq:rate constraint 2}) are active. This
suggests that
$r_1^*=\log\left|\mv{I}+\mv{H}_{11}\mv{S}_1^*\mv{H}_{11}^H\right|$,
i.e., SD is optimal. However, the optimal solution $\mv{S}_1^*$ of
(P2), or that of (P3) with $\mu_1=\mu_1^*$, denoted by
$\tilde{\mv{S}}_1^{\rm SD}$, in general does not have any
closed-form expression, and thus needs to be obtained by a numerical
search. Since (P3) is convex, the interior-point method \cite{Boyd}
can be used to efficiently obtain its solution for a given $\mu_1$.
Let $\mv{S}_1^{\star}(\mu_1)$ denote the optimal solution of (P3)
for a given $\mu_1$. Then, $\mu_1^*$ can be efficiently found by a
simple bisection search based upon the sub-gradient  \cite{Boyd} of
$g(\mu_1)$, which can be shown from (\ref{eq:Lagrangian new}) (with
$\mu_2=1-\mu_1$) to be
\begin{equation}
\log\left|\mv{I}+\left(\mv{I}+\mv{H}_{11}\mv{S}_1^{\star}(\mu_1)\mv{H}_{11}^H\right)^{-1}\mv{H}_{21}\mv{S}_2\mv{H}_{21}^H
\right|-r_2.
\end{equation}
Once $\mu_1$ converges to $\mu_1^*$, the corresponding
$\mv{S}_1^{\star}(\mu_1)$ becomes the optimal $\tilde{\mv{S}}_1^{\rm
SD}$. The optimal $r_1^*$ in this case with SD, denoted by
$\tilde{r}_1^{\rm SD}$, is then expressed as
\begin{equation}
\tilde{r}_1^{\rm
SD}=\log\left|\mv{I}+\mv{H}_{11}\tilde{\mv{S}}_1^{\rm
SD}\mv{H}_{11}^H\right|.
\end{equation}

Similarly like the previous two cases and using Lemma
\ref{lemma:inequality}, we can show that this case holds when
\begin{equation}
\hat{R}_2^{(a)}\leq r_2 \leq \bar{R}_2^{(a)}.
\end{equation}
\end{itemize}

\subsubsection{Combing $r_2>R_2^{(b)}$ and $r_2\leq R_2^{(b)}$}

To summarize, the following theorem is obtained for the optimal
solution of (P1).
\begin{theorem}
For a given set of $\mv{S}_2$ and $r_2$ of user 2, the optimal
transmit covariance matrix of user 1 and the maximum transmit rate
of user 1 with the proposed OMD are given as follows:
\begin{eqnarray}\label{eq:optimal sol S}
\mv{S}_1^{\rm OMD}=\left\{\begin{array}{ll} \mv{S}_1^{\rm SD}, & 0<r_2<\hat{R}_2^{(a)} \\
\tilde{\mv{S}}_1^{\rm SD}, & \hat{R}_2^{(a)} \leq r_2\leq \bar{R}_2^{(a)} \\
\mv{S}_1^{\rm JD}, & \bar{R}_2^{(a)}< r_2 \leq R_2^{(b)}  \\
\mv{S}_1^{\rm SUD}, & r_2>R_2^{(b)},
\end{array} \right.
\end{eqnarray}
\begin{eqnarray}\label{eq:optimal sol r}
r_1^{\rm OMD}=\left\{\begin{array}{ll} r_1^{\rm SD}, & 0< r_2< \hat{R}_2^{(a)} \\
\tilde{r}_1^{\rm SD}, & \hat{R}_2^{(a)} \leq r_2\leq \bar{R}_2^{(a)} \\
r_1^{\rm JD}, & \bar{R}_2^{(a)}< r_2 \leq R_2^{(b)}  \\
r_1^{\rm SUD}, & r_2>R_2^{(b)}.
\end{array} \right.
\end{eqnarray}
The corresponding optimal decoding methods at user 1's receiver are
(from top to bottom) SD, SD, JD, and SUD, respectively.
\end{theorem}

In Fig. \ref{fig:rate region}, we show $r_1^{\rm OMD}$ in
(\ref{eq:optimal sol r}) as a function of $r_2$ for some fixed
$\mv{S}_2$. The rate gain of $r_1^{\rm OMD}$ for OMD over $r_1^{\rm
SUD}$ for SUD is clearly shown when $r_2<R_2^{(b)}$. There are three
pentagon-shape capacity regions shown in the figure, which are
$\mathcal{C}_{\rm MAC}(\mv{S}_1^{\rm JD},\mv{S}_2)$,
$\mathcal{C}_{\rm MAC}(\mv{S}_1^{\rm SD},\mv{S}_2)$, and
$\mathcal{C}_{\rm MAC}(\tilde{\mv{S}}_1^{\rm SD},\mv{S}_2)$,
respectively, where $\mathcal{C}_{\rm MAC}(\mv{S}_1,\mv{S}_2)$
denotes the capacity region of a two-user Gaussian MIMO-MAC with
user 1's and user 2's transmitters transmitting to user 1's
receiver, and $\mv{S}_1$, $\mv{S}_2$ denoting the transmit
covariance matrices of user 1 and user 2, respectively. More
specifically, $\mathcal{C}_{\rm MAC}(\mv{S}_1,\mv{S}_2)$ can be
expressed as \cite{Cover}
\begin{align}\label{eq:MAC capacity region}
\mathcal{C}_{\rm MAC}(\mv{S}_1,\mv{S}_2)\triangleq\left\{(r_1,r_2):
\sum_{i\in\mathcal{J}}r_i\leq
\log\left|\mv{I}+\sum_{i\in\mathcal{J}}\mv{H}_{i1}\mv{S}_i\mv{H}_{i1}^H\right|,
\forall \mathcal{J}\subseteq\{1,2\} \right\}.
\end{align}
Note that in Fig. \ref{fig:rate region}, the sold line consisting of
different rate pairs of $(r_1^{\rm OMD}, r_2)$ constitute the
boundary rate pairs of the aforementioned capacity regions. Also
note that there is a curved part of this rate-pair line in the case
of $\hat{R}_2^{(a)}<r_2<\bar{R}_2^{(a)}$, where $r_1^{\rm OMD}$ is
equal to $\tilde{r}_1^{\rm SD}$ and is achievable by
$\tilde{\mv{S}}_1^{\rm SD}$, which is the solution of problem (P3)
for some given $\mu_1$, $0<\mu_1<1$.

\section{Extension to More Than Two Users} \label{sec:K users}

In this section, we extend the results obtained for the two-user
MIMO system to the general MU-MIMO system with more than two users,
i.e., $K>2$. Due to the symmetry, we consider only user 1's transmit
optimization over $\mv{S}_1$ to maximize transmit rate $r_1$, with
all the other users' transmit rates, $r_2,\ldots,r_K$, and transmit
covariance matrices, $\mv{S}_2,\ldots,\mv{S}_K$, being fixed.

To apply OMD at user 1's receiver, we need to first identify the
group of users whose signals are (jointly or successively) decodable
at user 1's receiver without the presence of user 1's own received
signal. We thus have the following definitions:
\begin{definition}\label{def:1}
A set $\mathcal{U}_1$, $\mathcal{U}_1\subseteq\{2,\ldots,K\}$, is
called a {\it decodable user set} for user 1, if the received
signals at user 1's receiver due to the users in $\mathcal{U}_1$ are
decodable without the presence of user 1's own received signal, by
treating the received signals from the other users in
$\overline{\mathcal{U}_1}$ as colored Gaussian noise, where
$\overline{\mathcal{U}_1}$ denotes the complementary set of
$\mathcal{U}_1$, i.e.,
$\mathcal{U}_1\bigcap\overline{\mathcal{U}_1}=\varnothing$ and
$\mathcal{U}_1\bigcup\overline{\mathcal{U}_1}=\{2,\ldots,K\}$. More
specifically, the transmit rates of users in $\mathcal{U}_1$ must
satisfy \cite{Cover}
\begin{align}
\sum_{i\in\mathcal{J}}r_i\leq\log\left|\mv{I}+\left(\mv{I}+\sum_{k\in\overline{\mathcal{U}_1}}\mv{H}_{k1}\mv{S}_k\mv{H}_{k1}^H\right)^{-1}
\sum_{i\in\mathcal{J}}\mv{H}_{i1}\mv{S}_i\mv{H}_{i1}^H\right|,
\forall \mathcal{J}\subseteq \mathcal{U}_1.
\end{align}
\end{definition}
\vspace{0.2in}
\begin{definition}
A set $\mathcal{U}_1^*\subseteq\{2,\ldots,K\}$ is called an {\it
optimal} decodable user set for user 1, if $\mathcal{U}_1^*$ is a
decodable user set for user 1, and among all possible decodable user
sets for user 1, $\mathcal{U}_1^*$  has the largest size.
\end{definition}

Next, we have the following important proposition:
\begin{proposition}\label{proposition}
The set $\mathcal{U}_1^*$ is unique. Furthermore, for any decodable
user set for user 1, $\mathcal{U}_1$, it holds that
$\mathcal{U}_1\subseteq\mathcal{U}_1^*$.
\end{proposition}

\begin{proof}
See Appendix \ref{appendix:proof proposition}.
\end{proof}

For conciseness, we show the algorithm to find the unique set for
user 1, $\mathcal{U}_1^*$, in Appendix \ref{appendix:algorithm}.

From Proposition \ref{proposition}, it follows that the optimal
decoding strategy for user 1's receiver is applying OMD to the users
in the set $\mathcal{U}_1^*$ (it may be possible that
$\mathcal{U}_1^*=\varnothing$), while taking the users in the set
$\overline{\mathcal{U}_1^*}$ as additional colored Gaussian noise.
For an arbitrary set $\mathcal{V}$, let $|\mathcal{V}|$ denote the
size of $\mathcal{V}$. Note that to make the OMD feasible, the rate
of user 1, $r_1$, and the rates of users in $\mathcal{U}_1^*$ must
be jointly in the capacity region of the corresponding
$(|\mathcal{U}_1^*|+1)$-user Gaussian MIMO-MAC for a given set of
user transmit covariance matrices and the receiver noise covariance
matrix,
$\mv{\Phi}=\mv{I}+\sum_{k\in\overline{\mathcal{U}_1^*}}\mv{H}_{k1}\mv{S}_k\mv{H}_{k1}^H$,
which, similar to (\ref{eq:MAC capacity region}), can be defined as
\begin{align}\label{eq:MAC capacity region new}
\mathcal{C}_{\rm
MAC}(\mathcal{U}_1^*)\triangleq&\bigg\{(r_1,\{r_i\}_{i\in
\mathcal{U}_1^*}):
\sum_{i\in\mathcal{J}}r_i\leq\log\left|\mv{I}+\mv{\Phi}^{-1}
\sum_{i\in\mathcal{J}}\mv{H}_{i1}\mv{S}_i\mv{H}_{i1}^H\right|,
\forall \mathcal{J}\subseteq \{1\}\bigcup\mathcal{U}_1^* \bigg\}.
\end{align}
Note that in (\ref{eq:MAC capacity region new}), the rate
inequalities involving subsets $\mathcal{J}$'s containing users
solely from $\mathcal{U}_1^*$ all hold due to the definition of
$\mathcal{U}_1^*$. Therefore, in order to find the optimal
$\mv{S}_1$ for user 1 to maximize $r_1$, with fixed $r_i$'s and
$\mv{S}_i$'s, $i=2,\ldots,K$, it is sufficient to consider the
following optimization problem:
\begin{align}
\mbox{(P4)}~~\mathop{\mathtt{max}}_{\mv{S}_1, r_1} & ~~~ r_1
\nonumber \\
\mathtt{s.t.} & ~~~
r_1+\sum_{i\in\mathcal{J}}r_i\leq\log\left|\mv{I}+\mv{\Phi}^{-1}
\left(\mv{H}_{11}\mv{S}_1\mv{H}_{11}^H+\sum_{i\in\mathcal{J}}\mv{H}_{i1}\mv{S}_i\mv{H}_{i1}^H\right)\right|,
\forall \mathcal{J}\subseteq \mathcal{U}_1^* \label{eq:rate constraints general case}\\
&~~~r_1\geq 0, \mathtt{tr}(\mv{S}_1)\leq P_1, \mv{S}_1\succeq 0
\end{align}

Problem (P4) is convex in terms of $r_1$ and $\mv{S}_1$ since its
constraints specify a convex set of $(r_1, \mv{S}_1)$. Similarly
like for problem (P2), we introduce a set of non-negative dual
variables, $\mu_n$'s, $n=1,\ldots, 2^{|\mathcal{U}_1^*|}$, each
associated with one corresponding constraint in (\ref{eq:rate
constraints general case}) for a particular subsect $\mathcal{J}$
(including $\mathcal{J}=\varnothing$) denoted by $\mathcal{J}_n$,
and obtain an equivalent problem for the optimization over
$\mv{S}_1$ for a given set of fixed $\mu_n$'s, which is expressed as
\begin{align}
\mbox{(P5)}~~\mathop{\mathtt{max}}_{\mv{S}_1} & ~~~ \sum_{n=1}^{
2^{|\mathcal{U}_1^*|}}\mu_n\log\left|\mv{I}+\mv{\Phi}^{-1}
\left(\mv{H}_{11}\mv{S}_1\mv{H}_{11}^H+\sum_{i\in\mathcal{J}_n}\mv{H}_{i1}\mv{S}_i\mv{H}_{i1}^H\right)\right|
\nonumber \\
\mathtt{s.t.} & ~~~ \mathtt{tr}(\mv{S}_1)\leq P_1, \mv{S}_1\succeq
0.
\end{align}

It can be shown that problem (P5) is convex, and thus it can be
solved via standard convex optimization techniques, e.g., the
interior point method \cite{Boyd}, while in general, no closed-form
solution for (P5) is available, similar to the previous two-user
case in Section \ref{sec:2 users}. Let the optimal solution of (P5)
be denoted by $\mv{S}_1^{\star}(\{\mu_n\})$. Then, $\mu_n$'s can be
updated towards the optimal dual solutions of (P4) via the
well-known ellipsoid method \cite{Boyd} subject to an additional
constraint, $\sum_n\mu_n=1$ (similar to Lemma \ref{lemma:optimal mu}
in the two-user case). Let the optimal solutions of $\mu_n$'s be
denoted by $\mu_n^*$'s. The optimal solution of $\mv{S}_1$ for (P4)
with OMD is then obtained as $\mv{S}_1^{\rm
OMD}=\mv{S}_1^{\star}(\{\mu_n^*\})$, and the corresponding maximum
achievable rate of user 1, $r_1^{\rm OMD}$, can be obtained from any
active constraint in (\ref{eq:rate constraints general case}) with
equality. The optimal decoding orders/decoding methods for the users
in $\mathcal{U}_1^*$ prior to decoding user 1's message can be
obtained according to the optimal non-zero dual solutions,
$\mu_n^*$'s, or equivalently, the corresponding active constraints
in (\ref{eq:rate constraints general case}) with equality, via
applying the property of polymatroid structure of $\mathcal{C}_{\rm
MAC}(\mathcal{U}_1^*)$ given in (\ref{eq:MAC capacity region new})
\cite{Tse98}.

\section{Simulation Results} \label{sec:numerical results}

In this section, the performance of the proposed OMD is evaluated in
comparison with the conventional SUD in a decentralized MU-MIMO
system with $K=2$ users, where the two users adopt an IWF-like
algorithm to successively in turn optimize their transmit covariance
matrices for individual rate maximization by deploying OMD or SUD at
their receivers. For the purpose of exposition, all the channels
involved in the system, including user's direct-link and cross-link
channels, are assumed to have independent Rayleigh-fading
distributions, i.e., each element of the channel matrix is
independent and identically distributed as zero-mean CSCG random
variable. Furthermore, each element of the two users' direct-link
channels is assumed to have the variance  $\rho_{11}$ and
$\rho_{22}$, for user 1 and 2, respectively; and each element of the
two cross-link channels has the variance, $\rho_{12}$ for the
channel from user 1 to user 2 and $\rho_{21}$ for the channel from
user 2 to user 1, respectively. In total, 5000 independent channel
realizations are simulated over which each user's achievable average
rate is computed. For each channel realization, the two users
iteratively update their transmit covariance matrices until their
rates both get converged. It is assumed that $M_k=N_k=2, k=1,2$.

In Fig. \ref{fig:sum rate}, the achievable average sum-rate of the
two users is shown for a symmetric system and channel setup, where
$P_1=P_2=100$, $\rho_{11}=\rho_{22}=1$, and
$\rho_{12}=\rho_{21}=\rho$. The user sum-rate is plotted against
$\rho$ to investigate the effect of the interference between the two
users on their achievable sum-rate. It is observed that the sum-rate
with the proposed OMD improves over that with the conventional SUD
for all the values of $\rho$, while the rate gains become more
substantial in the case of large values of $\rho$, i.e., the
``strong'' interference case. With SUD, it is observed that the
sum-rate first decreases with increasing of $\rho$ (as a result of
interference whitening), and then starts to increase with $\rho$ (as
a result of interference avoidance), and finally gets converged for
large values of $\rho$ (due to the fact that zero-forcing (ZF)
-based receive beamforming to completely null the co-channel
interference becomes optimal at the high signal-to-noise ratio (SNR)
region). However, the sum-rate with the proposed OMD is observed to
increase consistently with $\rho$, due to the fact that when the
co-channel interference becomes stronger at the receiver, the OMD
more easily decodes the interference.

Next, we consider a special scenario of the general system model
studied in this paper. In this case, a ``cognitive radio (CR)'' type
of newly emerging wireless system is considered, where user 1 is the
so-called primary (non-cognitive) user (PU) who is the legitimate
user operating in the frequency band of interest, while user 2 is
the secondary (cognitive) user (SU) that transmits simultaneously
with the PU over the same spectrum under the constraint that its
transmission will not cause the PU's transmission performance to an
unacceptable level \cite{Zhang08}. The PU is non-cognitive since it
is oblivious to the existence of the SU and applies the conventional
SUD at the receiver by treating the interference from the SU as
additional noise. While for the SU, it is cognitive in the sense
that it is aware of the PU and thus transmits with a much lower
average power than that of the PU in order to protect the PU; thus,
for this example it is assumed that $P_1=10P$ and $P_2=P$, where $P$
is a given constant. In addition, since the SU is cognitive, it may
choose to use the more advanced OMD at the receiver to cope with the
interference from the PU. Two cases are thus studied for this
example: Case (I) both user 1 and user 2 employ SUD; and Case (II)
user 1 employs SUD while user 2 employs OMD. It is assumed that the
SU's link distance is much shorter than that of the PU link, and
furthermore the SU transmitter and receiver are both in the vicinity
of the PU transmitter while they are both sufficiently far away from
the PU receiver. Thus, for this example we assume that $\rho_{11}=1,
\rho_{22}=10, \rho_{12}=10$, and $\rho_{21}=1$.

In Fig. \ref{fig:CR rate}, the achievable user individual rates are
shown for different values of $P$ in both Cases I and II. It is
observed that the achievable rate of user 2 (the SU) improves
significantly in Case II over Case I, thanks to the use of OMD
instead of SUD. This rate gain is substantial because the SU
receiver is close to the PU transmitter and thus $\rho_{12}$ is
large, i.e., the cross-link channel from PU to SU is a ``strong''
interference channel, for which the OMD is crucial for the SU to
mitigate the PU's interference. However, it is also observed that
the achievable rate of user 1 (the PU) drops slightly in Case II as
compared with Case I. This is because that in Case II with OMD, the
SU's transmitted signal has a more spatially spread-out spectrum
than that in Case I with SUD, and so does the received SU's
interference at the PU receiver. Nevertheless, due to the small
value of $\rho_{21}$ or the weak cross-link channel from SU to PU,
the capacity loss of the PU is not significant, which justifies the
operation principle of the SU, i.e., the PU transmission should be
sufficiently protected.

\section{Conclusion} \label{sec:conclusion}

This paper studied a new decoding method, namely opportunistic
multiuser detection (OMD), for the decentralized MU-MIMO system
where each user iteratively optimizes transmit covariance matrix for
individual rate maximization. In comparison with the conventional
single-user detection (SUD), the proposed OMD still allows a fully
decentralized processing of each user in the system, while it
improves the user's interference mitigation capability at the
receiver, and leads to more optimum spatial spectrum sharing among
the users. Simulation results showed that substantial system
throughput gains could be achieved by the proposed OMD over the
conventional SUD, for certain application scenarios.

\appendices

\section{Proof of Lemma \ref{lemma:optimal mu}}\label{appendix:proof optimal mu}

We will prove Lemma \ref{lemma:optimal mu} by contradiction. First,
suppose that $\mu_1^*+\mu_2^*<1$. Then, in the maximization problem
of (\ref{eq:Lagrange dual}), from the expression of
$\mathcal{L}(r_1,\mv{S}_1,\mu_1,\mu_2)$ in (\ref{eq:Lagrangian
new}), it follows that the optimal $r_1$ that maximizes the
Lagrangian is $r_1^*=+\infty$, which contradicts the fact that $r_1$
in (P2) is upper-bounded by finite rate values in the constraints
(\ref{eq:rate constraint 1}) and (\ref{eq:rate constraint 2}).
Second, suppose that $\mu_1^*+\mu_2^*>1$. Similarly like the
previous case, it can shown that $r_1^*=0$. However, this can not be
true since we can easily find a feasible solution set for
$(r_1,\mv{S}_1)$ in (P2) such that $r_1>0$. By combining the above
two cases, it follows that $\mu_1^*+\mu_2^*=1$.

\section{Proof of Lemma \ref{lemma:inequality}}\label{appendix:proof inequality}

We rewrite $\bar{R}_2^{(a)}$ in (\ref{eq:R(1) bar}) and
$\hat{R}_2^{(a)}$ in (\ref{eq:R(3)}) as
\begin{align}
\bar{R}_2^{(a)}&=\log\left|\mv{I}+\mv{H}_{11}\mv{S}_1^{\rm
JD}\mv{H}_{11}^H+\mv{H}_{21}\mv{S}_2\mv{H}_{21}^H
\right|-\log\left|\mv{I}+\mv{H}_{11}\mv{S}_1^{\rm
JD}\mv{H}_{11}^H\right| \label{eq:inequality 1}\\
\hat{R}_2^{(a)}&=\log\left|\mv{I}+\mv{H}_{11}\mv{S}_1^{\rm
SD}\mv{H}_{11}^H+\mv{H}_{21}\mv{S}_2\mv{H}_{21}^H
\right|-\log\left|\mv{I}+\mv{H}_{11}\mv{S}_1^{\rm
SD}\mv{H}_{11}^H\right|. \label{eq:inequality 2}
\end{align}
Since $\mv{S}_1^{\rm JD}$ and $\mv{S}_1^{\rm SD}$ are optimal for
the sum-capacity (in an equivalent two-user MIMO-MAC) and user'1
channel capacity (without the presence of user 2), respectively, we
have
\begin{align}
\log\left|\mv{I}+\mv{H}_{11}\mv{S}_1^{\rm
JD}\mv{H}_{11}^H+\mv{H}_{21}\mv{S}_2\mv{H}_{21}^H \right|&\geq
\log\left|\mv{I}+\mv{H}_{11}\mv{S}_1^{\rm
SD}\mv{H}_{11}^H+\mv{H}_{21}\mv{S}_2\mv{H}_{21}^H \right| \\
\log\left|\mv{I}+\mv{H}_{11}\mv{S}_1^{\rm
JD}\mv{H}_{11}^H\right|&\leq
\log\left|\mv{I}+\mv{H}_{11}\mv{S}_1^{\rm SD}\mv{H}_{11}^H\right|.
\end{align}
Combining the above two inequalities with (\ref{eq:inequality 1})
and (\ref{eq:inequality 2}), it thus follows that
$\bar{R}_2^{(a)}\geq \hat{R}_2^{(a)}$.

\section{Proof of Proposition \ref{proposition}}\label{appendix:proof proposition}

We first prove the former part of Proposition \ref{proposition},
i.e., the set $\mathcal{U}_1^*$ is unique, by contradiction. Suppose
that there exist two optimal decodable user sets for user 1 with the
same size, denoted by $\mathcal{A}_1$ and $\mathcal{B}_1$. Without
loss of generality, we let $\mathcal{A}_1=\{\mathcal{D},
\mathcal{C}\}$ and $\mathcal{B}_1=\{\mathcal{E}, \mathcal{C}\}$,
where $\mathcal{C}$, $\mathcal{D}$ and $\mathcal{E}$ are subsets
consisting of completely different user indexes. Then, we can
express $\overline{\mathcal{A}}_1=\{\mathcal{E},\mathcal{F}\}$ and
$\overline{\mathcal{B}}_1=\{\mathcal{D},\mathcal{F}\}$, where
$\mathcal{F}=\overline{\mathcal{A}_1\bigcup\mathcal{B}_1}$. Then,
for users in the set $\mathcal{A}_1$, their transmit rates must
satisfy \cite{Cover}
\begin{align}\label{eq:rate J}
\sum_{i\in\mathcal{J}\bigcup\mathcal{K}}r_i\leq\log\left|\mv{I}+\left(\mv{I}+\sum_{k\in\overline{\mathcal{A}_1}}\mv{H}_{k1}\mv{S}_k\mv{H}_{k1}^H\right)^{-1}
\sum_{i\in\mathcal{J}\bigcup\mathcal{K}}\mv{H}_{i1}\mv{S}_i\mv{H}_{i1}^H\right|,
\forall \mathcal{J}\subseteq \mathcal{D}, \mathcal{K}\subseteq
\mathcal{C}.
\end{align}
Similarly, for users in the subset $\mathcal{E}$ of $\mathcal{B}_1$,
their transmit rates must satisfy
\begin{align}\label{eq:rate I}
\sum_{i\in\mathcal{I}}r_i\leq\log\left|\mv{I}+\left(\mv{I}+\sum_{k\in\overline{\mathcal{B}_1}}\mv{H}_{k1}\mv{S}_k\mv{H}_{k1}^H\right)^{-1}
\sum_{i\in\mathcal{I}}\mv{H}_{i1}\mv{S}_i\mv{H}_{i1}^H\right|,
\forall \mathcal{I}\subseteq \mathcal{E}.
\end{align}
Let $\mathcal{J}'$ be an orthogonal set of $\mathcal{J}$, where
$\mathcal{J}'\bigcup\mathcal{J}=\mathcal{D}$. Similarly,
$\mathcal{I}'$ is defined for $\mathcal{I}$, where
$\mathcal{I}'\bigcup\mathcal{I}=\mathcal{E}$. (\ref{eq:rate J}) and
(\ref{eq:rate I}) can thus be further shown as follows:
\begin{align}\label{eq:rate J new}
\sum_{i\in\mathcal{J}\bigcup\mathcal{K}}r_i\leq\log\left|\mv{I}+\left(\mv{I}+\sum_{k\in\mathcal{I}\bigcup\mathcal{F}}\mv{H}_{k1}\mv{S}_k\mv{H}_{k1}^H\right)^{-1}
\sum_{i\in\mathcal{J}\bigcup\mathcal{K}}\mv{H}_{i1}\mv{S}_i\mv{H}_{i1}^H\right|
\end{align}
\begin{align}\label{eq:rate I new}
\sum_{i\in\mathcal{I}}r_i\leq\log\left|\mv{I}+\left(\mv{I}+\sum_{k\in\mathcal{J}\bigcup\mathcal{F}}\mv{H}_{k1}\mv{S}_k\mv{H}_{k1}^H\right)^{-1}
\sum_{i\in\mathcal{I}}\mv{H}_{i1}\mv{S}_i\mv{H}_{i1}^H\right|.
\end{align}
From (\ref{eq:rate J new}) and (\ref{eq:rate I new}), we obtain
\begin{align}\label{eq:rate I J}
\sum_{i\in\mathcal{J}\bigcup\mathcal{K}\bigcup\mathcal{I}}r_i\leq&\log\left|\mv{I}+
\sum_{i\in\mathcal{J}\bigcup\mathcal{K}\bigcup\mathcal{I}\bigcup\mathcal{F}}\mv{H}_{i1}\mv{S}_i\mv{H}_{i1}^H\right|+\log\left|\mv{I}+
\sum_{i\in\mathcal{I}\bigcup\mathcal{J}\bigcup\mathcal{F}}\mv{H}_{i1}\mv{S}_i\mv{H}_{i1}^H\right|
\nonumber \\ &-\log\left|\mv{I}+
\sum_{i\in\mathcal{I}\bigcup\mathcal{F}}\mv{H}_{i1}\mv{S}_i\mv{H}_{i1}^H\right|-\log\left|\mv{I}+
\sum_{i\in\mathcal{J}\bigcup\mathcal{F}}\mv{H}_{i1}\mv{S}_i\mv{H}_{i1}^H\right|
\end{align}
Since
\begin{align}\label{eq:rate I J new}
\log\left|\mv{I}+
\sum_{i\in\mathcal{I}\bigcup\mathcal{J}\bigcup\mathcal{F}}\mv{H}_{i1}\mv{S}_i\mv{H}_{i1}^H\right|
-\log\left|\mv{I}+
\sum_{i\in\mathcal{I}\bigcup\mathcal{F}}\mv{H}_{i1}\mv{S}_i\mv{H}_{i1}^H\right|\nonumber
\\ \leq\log\left|\mv{I}+
\sum_{i\in\mathcal{J}\bigcup\mathcal{F}}\mv{H}_{i1}\mv{S}_i\mv{H}_{i1}^H\right|-\log\left|\mv{I}+
\sum_{i\in\mathcal{F}}\mv{H}_{i1}\mv{S}_i\mv{H}_{i1}^H\right|
\end{align}
From (\ref{eq:rate I J}) and (\ref{eq:rate I J new}), it follows
that
\begin{align}\label{eq:rate I J final}
\sum_{i\in\mathcal{J}\bigcup\mathcal{K}\bigcup\mathcal{I}}r_i\leq&\log\left|\mv{I}+
\sum_{i\in\mathcal{J}\bigcup\mathcal{K}\bigcup\mathcal{I}\bigcup\mathcal{F}}\mv{H}_{i1}\mv{S}_i\mv{H}_{i1}^H\right|-\log\left|\mv{I}+
\sum_{i\in\mathcal{F}}\mv{H}_{i1}\mv{S}_i\mv{H}_{i1}^H\right|
\nonumber \\ =&
\log\left|\mv{I}+\left(\mv{I}+\sum_{k\in\mathcal{F}}\mv{H}_{k1}\mv{S}_k\mv{H}_{k1}^H\right)^{-1}
\sum_{i\in\mathcal{J}\bigcup\mathcal{K}\bigcup\mathcal{I}}\mv{H}_{i1}\mv{S}_i\mv{H}_{i1}^H\right|.
\end{align}
Thus, the set $\mathcal{J}\bigcup\mathcal{K}\bigcup\mathcal{I}$ is a
decodable user set for user 1 for any $\mathcal{J}\subseteq
\mathcal{D}, \mathcal{K}\subseteq \mathcal{C}$, and
$\mathcal{I}\subseteq \mathcal{E}$, and so is the set
$\mathcal{G}_1=\mathcal{D}\bigcup\mathcal{C}\bigcup\mathcal{E}$.
Since the size of $\mathcal{G}_1$ is larger than that of
$\mathcal{A}_1$ or $\mathcal{B}_1$, this contradicts the assumption
that $\mathcal{A}_1$ and $\mathcal{B}_1$ are optimal decodable user
sets for user 1. The proof of the former part of Proposition
\ref{proposition} thus follows.

Next, we prove the latter part of Proposition \ref{proposition},
i.e., any decodable user set for user 1, $\mathcal{U}_1$, must be a
subset of $\mathcal{U}_1^*$. The proof is also obtained via
contradiction. Suppose that there is a set $\mathcal{U}_1$ that is
not a subset of $\mathcal{U}_1^*$. Without loss of generality, we
can express $\mathcal{U}_1=\{\mathcal{D}, \mathcal{C}\}$ and
$\mathcal{U}_1^*=\{\mathcal{E}, \mathcal{C}\}$, where $\mathcal{C}$,
$\mathcal{D}$ and $\mathcal{E}$ are orthogonal subsets. Based on the
proof for the former part of Proposition \ref{proposition}, we know
that the set $\mathcal{D}\bigcup\mathcal{C}\bigcup\mathcal{E}$ is
also a decodable user set for user 1, and apparently, it has a
larger size than $\mathcal{U}_1^*$, which contradicts the fact that
$\mathcal{U}_1^*$ is the optimal decodable user set for user 1. The
proof of the latter part of Proposition \ref{proposition} thus
follows.

\begin{table}
\centering
\begin{tabular}{|l|}
\hline \hspace*{0.0cm} Initialize $\mathcal{V}=\{2,\ldots,K\}$,
$\overline{\mathcal{V}}=\varnothing$. \\
\hspace*{0.0cm} While $|\mathcal{V}|>0$ do ~~~~~~~~~~~~(1)\\
\hspace*{0.5cm} Initialize $n=1$ \\
\hspace*{0.5cm} While $n\leq 2^{|\mathcal{V}|}-1$ do\\
\hspace*{1.0cm} If $\sum_{i\in\mathcal{V}_n}r_i\leq C(\mathcal{V}_n)$ \\
\hspace*{1.5cm} Set $n\leftarrow n+1$ \\
\hspace*{1.0cm} Else \\
\hspace*{1.5cm} Set $\mathcal{V}\leftarrow \mathcal{V}-\mathcal{V}_{n}$ \\
\hspace*{1.5cm} Set $\overline{\mathcal{V}}\leftarrow \overline{\mathcal{V}}\bigcup\mathcal{V}_{n}$\\
\hspace*{1.5cm} Go to (1) \\
\hspace*{1.0cm} End If \\
\hspace*{0.5cm} End While\\
\hspace*{0.5cm} Go to (2)\\
\hspace*{0.0cm} End While\\
\hspace*{0.0cm} Set $\mathcal{U}_1^*=\mathcal{V}$. ~~~~~~~~~~~~~~~~~ (2) \\
\hline
\end{tabular}
\caption{The algorithm to find $\mathcal{U}_1^*$.} \label{table}
\end{table}

\section{Algorithm to Find $\mathcal{U}_1^*$}\label{appendix:algorithm}

In this appendix, we present an algorithm to find the optimal
decodable user set for user 1, $\mathcal{U}_1^*$. First, some
notations are given as follows for the convenience of presentation.
Let $\mathcal{V}_n$ denote a subset of an arbitrary set
$\mathcal{V}$, $n=1,\ldots,2^{|\mathcal{V}|}-1$. Note that here we
have excluded the case that $\mathcal{V}_n=\varnothing$ for the ease
of presentation. The operation $\mathcal{V}-\mathcal{V}_n$ then
stands for removing the subset $\mathcal{V}_n$ from $\mathcal{V}$.

For a given user set, $\mathcal{V}\subseteq\{2,\ldots,K\}$, we know
from Definition \ref{def:1} that $\mathcal{V}$ is a decodable user
set for user 1 if and only if for any subset of $\mathcal{V}$,
$\mathcal{V}_n$, it satisfies that
\begin{align}
\sum_{i\in\mathcal{V}_n}r_i\leq\log\left|\mv{I}+\left(\mv{I}+\sum_{k\in\overline{\mathcal{V}}}\mv{H}_{k1}\mv{S}_k\mv{H}_{k1}^H\right)^{-1}
\sum_{i\in\mathcal{V}_n}\mv{H}_{i1}\mv{S}_i\mv{H}_{i1}^H\right|\triangleq
C(\mathcal{V}_n).
\end{align}
However, if there exists a subset $\mathcal{V}_n$ such that
$\sum_{i\in\mathcal{V}_n}r_i>C(\mathcal{V}_n)$, it follows that
$\mathcal{V}$ should not be a decodable user set for user 1. From
the above property, we are able to design an iterative algorithm to
find $\mathcal{U}_1^*$, which is explained as follows. Initially, we
let $\mathcal{V}=\{2,\ldots,K\}$. Thus,
$\overline{\mathcal{V}}=\varnothing$. Then, we will sequentially
check for all the subsets of $\mathcal{V}$ whether
$\sum_{i\in\mathcal{V}_n}r_i\leq C(\mathcal{V}_n), \forall n$. If
this is the case, then we declare that
$\mathcal{U}_1^*=\mathcal{V}$. However, if we find any $n'$ such
that $\sum_{i\in\mathcal{V}_{n'}}r_i>C(\mathcal{V}_{n'})$, then we
conclude that $\mathcal{V}$ should not be $\mathcal{U}_1^*$ and
furthermore $\mathcal{U}_1^*\subseteq \mathcal{V}-\mathcal{V}_{n'}$.
In this case, we will set $\mathcal{V}\leftarrow
\mathcal{V}-\mathcal{V}_{n'}$, $\overline{\mathcal{V}}\leftarrow
\overline{\mathcal{V}}\bigcup\mathcal{V}_{n'}$, and start a new
sequence of tests for $\sum_{i\in\mathcal{V}_n}r_i\leq
C(\mathcal{V}_n), \forall n$. The above procedure iterates until we
find a set $\mathcal{V}$ such that $\sum_{i\in\mathcal{V}_n}r_i\leq
C(\mathcal{V}_n), \forall n$ or $\mathcal{V}=\varnothing$. In both
cases, we set $\mathcal{U}_1^*=\mathcal{V}$. The above algorithm is
summarized in Table \ref{table}.

\newpage

\begin{figure}
\psfrag{a}{$R_2^{(b)}$}\psfrag{b}{$\bar{R}_2^{(a)}$}\psfrag{c}{$\hat{R}_2^{(a)}$}\psfrag{d}{$r_1^{\rm
SUD}$}\psfrag{e}{{\small $C_{\rm MAC}(\mv{S}_1^{\rm JD},\mv{S}_2)$}}
\psfrag{f}{{\small $C_{\rm MAC}(\mv{S}_1^{\rm SD},\mv{S}_2)$}}
\psfrag{g}{{\small $C_{\rm MAC}(\tilde{\mv{S}}_1^{\rm
SD},\mv{S}_2)$}} \psfrag{h}{{\small $r_1^{\rm OMD}$}}
\psfrag{i}{$r_1^{\rm JD}$}\psfrag{l}{$\tilde{r}_1^{\rm
SD}$}\psfrag{m}{$r_1^{\rm SD}$}
\begin{center}
\scalebox{0.9}{\includegraphics*[70pt,200pt][540pt,590pt]{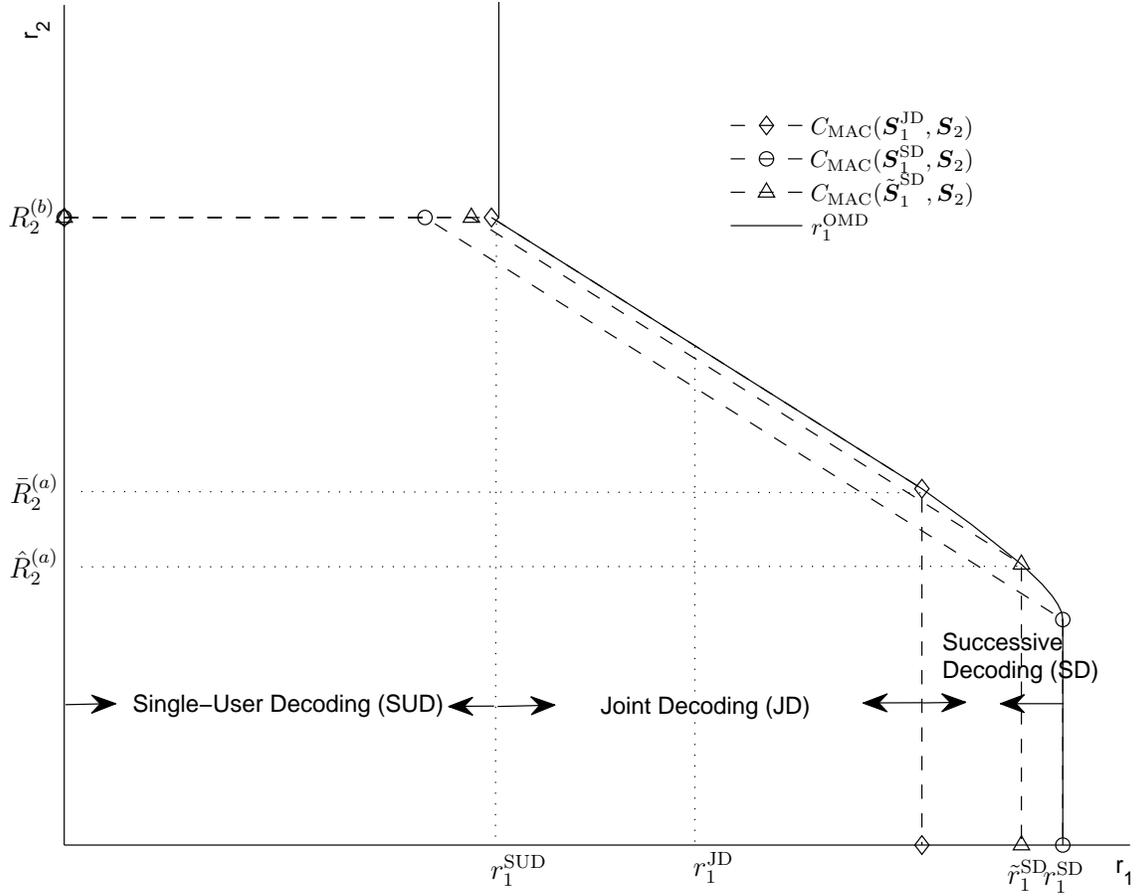}}
\end{center}
\caption{The maximum achievable rate of user 1 with OMD, $r_1$, as a
function of user 2's rate, $r_2$, for some fixed
$\mv{S}_2$.}\label{fig:rate region}
\end{figure}

\begin{figure}
\centering{
 \epsfxsize=5in
    \leavevmode{\epsfbox{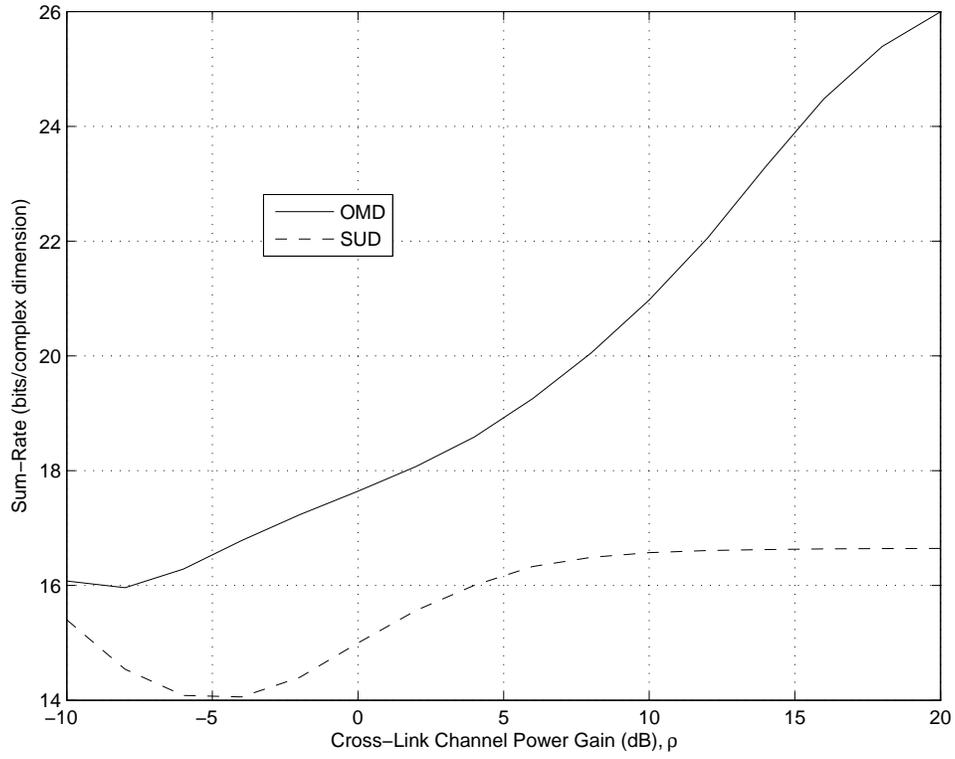}}}
\caption{The achievable sum-rate versus the average cross-link
channel power gain, $\rho$, for a MU-MIMO system with $K=2$,
$M_k=N_k=2, k=1,2$, and $P_1=P_2=100$.}\label{fig:sum rate}
\end{figure}

\begin{figure}
\centering{
 \epsfxsize=5in
    \leavevmode{\epsfbox{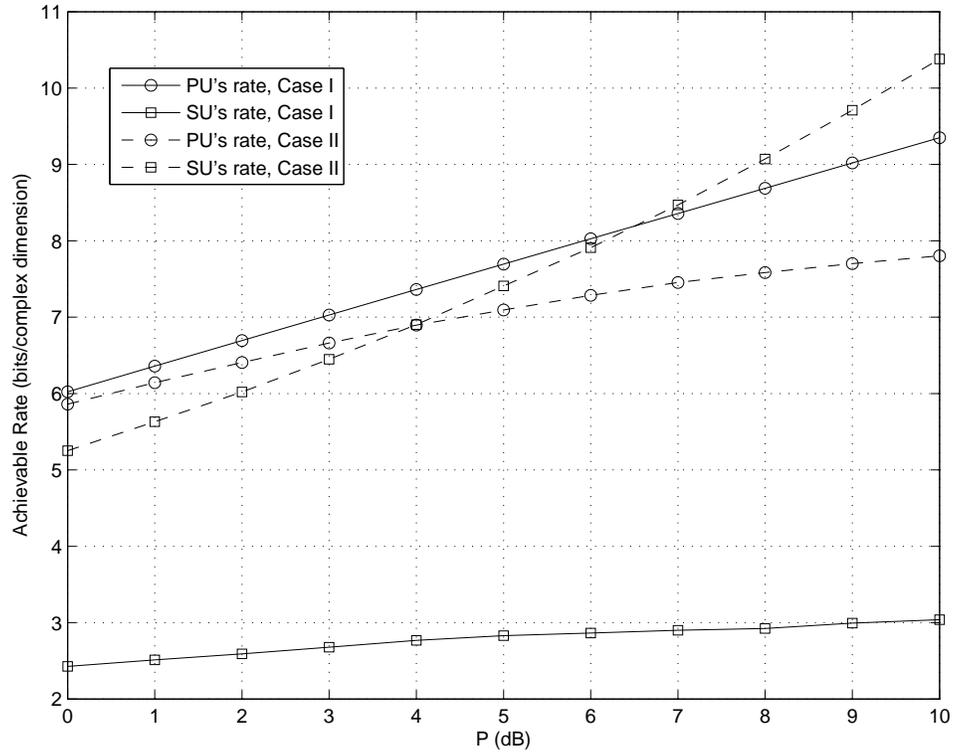}}}
\caption{The achievable rate versus the average transmit power, $P$,
in a MIMO CR system with $M_k=N_k=2, k=1,2$, $P_1=10P$, and $P_2=P$,
for different decoding methods: Case (I) both PU and SU employ SUD;
and Case (II) PU employs SUD and SU employs OMD.}\label{fig:CR rate}
\end{figure}
\end{document}